\documentclass{article}       
\AtBeginDocument{%
  \providecommand\BibTeX{{%
    \normalfont B\kern-0.5em{\scshape i\kern-0.25em b}\kern-0.8em\TeX}}}

\usepackage{hyperref}
\usepackage{amsmath}
\usepackage{amsthm}
\newtheorem{theorem}{Theorem}[section]

\newtheorem{lemma}[theorem]{Lemma}

\usepackage{amssymb}
\usepackage{listings}
\usepackage{color}
\usepackage{algpseudocode}
\usepackage{algorithm}
\usepackage{array}
\usepackage{breakurl}
\usepackage{multirow}
\usepackage{rotating}
\usepackage{indentfirst}
\usepackage{mdwlist}
\usepackage{setspace}

\usepackage{booktabs}

\usepackage{graphicx}
\usepackage{color}
\usepackage{subcaption}

\usepackage[version=4]{mhchem}

\usepackage[T1]{fontenc}

\usepackage{mathtools}

\providecommand{\keywords}[1]
{
  \small	
  \textbf{Keywords:} #1
}

\title{Optimally selecting the top $k$ values from $X+Y$ with layer-ordered heaps}

\author{Oliver Serang\\
University of Montana \\ Dept. of Computer Science \\ Missoula, MT, USA}

\date{\today}

\begin{document}

\maketitle

\begin{abstract}
  \noindent Selection and sorting the Cartesian sum, $X+Y$, are
  classic and important problems. Here, a new algorithm is presented,
  which generates the top $k$ values of the form $X_i+Y_j$. The
  algorithm relies only on median-of-medians and is simple to
  implement. Furthermore, it uses data structures contiguous in
  memory, and is fast in practice. The presented algorithm is
  demonstrated to be theoretically optimal.
\end{abstract}


\section{Introduction}
Given two vectors of length $n$, $X$ and $Y$, top-$k$ on $X+Y$ finds
the $k$ smallest values of the form $X_i+Y_j$. Note that this problem
definition is presented w.l.o.g.; $X$ and $Y$ need not share the same
length. Top-$k$ is important to practical applications, such as
selecting the most abundant $k$ isotope peaks from a
compound\cite{kreitzberg:fast}. Top-$k$ is $\in\Omega(n+k)$, because
loading the vectors is $\in\Theta(n)$ and returning the minimal $k$
values is $\in\Theta(k)$.

\subsection{Existing, tree-based methods for top-$k$}
In 1982, Frederickson \& Johnson introduced a method reminiscent of
median-of-medians\cite{blum:time}; their method selects only the
$k^{\mbox{\footnotesize th}}$ value from $X+Y$ in $O(n +
\min(n,k)\log(\frac{k}{\min(n,k)}))
$ steps\cite{frederickson:complexity}.

Frederickson subsequently published a second algorithm, which finds the
$k$ smallest elements from a min-heap in $O(k)$, assuming the heap has
already been built\cite{frederickson:optimal}. Combining this method
with a combinatoric heap on $X+Y$ (described below for the Kaplan
\emph{et al.} method) solves top-$k$ in $O(n+k)$. Frederickson's
method uses a tree data structure similar to what would in 2000 be
formalized into Chazelle's soft heap\cite{chazelle:soft}, and can be
combined with a combinatoric heap to compute the
$k^{\mbox{\footnotesize th}}$ smallest value from $X+Y$.

Kaplan \emph{et al.} described an alternative method for selecting the
$k^{\mbox{\footnotesize th}}$ smallest value\cite{kaplan:selection};
that method explicitly used Chazelle's soft
heaps\cite{chazelle:soft}. By heapifying $X$ and $Y$ in linear time
(\emph{i.e.}, guaranteeing w.l.o.g. that $X_i \leq X_{2 i}, X_{2 i +
  1}$), $\min_{i,j} X_i+Y_j = X_1 + Y_1$.  Likewise, $X_i+Y_j \leq
X_{2 i} + Y_j, X_{2 i+1} + Y_j, X_i + Y_{2 j}, X_i + Y_{2 j + 1}$. The
soft heap is initialized to contain tuple $(X_1+Y_1,1,1)$. Then, as
tuple $(v,i,j)$ is popped from soft heap, lower-quality tuples are
inserted into the soft heap. These lower-quality tuples of $(i,j)$ are
\begin{equation}
\label{eqn:kaplan-proposal}
\begin{cases}
  \{(2i,1), (2i+1,1), (i,2), (i,3)\}, & j=1\\
  \{(i,2j), (i,2j+1)\}, & j>1.\\
\end{cases}
\end{equation}
In the matrix $X_i+Y_j$ (which is not realized), this scheme
progresses in row-major order, thereby avoiding a tuple being added
multiple times.

To compute the $k^{\mbox{\footnotesize th}}$ smallest value from
$X+Y$, the best $k$ values are popped from the soft heap. Even though
only the minimal $k$ values are desired, ``corruption'' in the soft
heap means that the soft heap will not always pop the minimal value;
however, as a result, soft heaps can run faster than the $\Omega(n
\log(n))$ lower bound on comparison sorting. $\epsilon \in
(0,\frac{1}{2}]$ is a free parameter to the soft heap. It bounds the
  number of corrupted elements in the soft heap (which may be promoted
  earlier in the queue than they should be) as $\leq t \cdot
  \epsilon$, where $t$ is the number of insertions into the soft heap
  thus far. Thus, instead of popping $k$ items (and inserting their
  lower-quality dependents as described in
  equation~\ref{eqn:kaplan-proposal}), the total number of pops $p$
  can be found: The maximal size of the soft heap after $p$ pops is
  $\leq 3p$ (because each pop removes one element and inserts $\leq 4$
  elements according to equation~\ref{eqn:kaplan-proposal});
  therefore, $p - corruption \geq p - 4 p\cdot \epsilon$, and thus $p
  - 4 p\cdot \epsilon \geq k$ guarantees that $p - corruption \geq
  k$. This leads to $p = \frac{k}{1-4\epsilon}$,
  $\epsilon\leq\frac{1}{4}$. This guarantees that $\Theta(k)$ values,
  which must include the minimal $k$ values, are popped. These values
  are post-processed to retrieve the minimal $k$ values via linear
  time one-dimensional selection\cite{blum:time}. For constant
  $\epsilon$, both pop and insertion operations to the soft heap are
  $\in \~{O}(1)$, and thus the overall runtime of the algorithm is
  $\in O(n + k)$.

Note that the Kaplan \emph{et al.} method easily solves top-$k$ in
$O(n+k)$ steps; this is because computing the $k^{\mbox{\footnotesize
    th}}$ smallest value from $X+Y$ pops the minimal $k$ values from
the soft heap.

\subsection{Layer-ordered heaps and a novel selection algorithm on $X+Y$}
This paper uses layer-ordered heaps (LOHs)\cite{kreitzberg:selection}
to produce an optimal selection algorithm on $X+Y$. LOHs are stricter
than heaps but not as strict as sorting: Heaps guarantee only that
$X_i \leq X_{child(i)}$, but do not guarantee any ordering between one
child of $X_i$, $a$, and the child of the sibling of $a$. Sorting is
stricter still, but sorting $n$ values cannot be done faster than
$\log_2(n!) \in \Omega(n \log(n))$. LOHs partition the array into
several layers such that the values in a layer are $\leq$ to the
values in subsequent layers: $X^{(u)} = X^{(u)}_1, X^{(u)}_2, \ldots
\leq X^{(u+1)}$. The size of these layers starts with $X^{(1)}=1$ and
grows exponentially such that $\lim\limits_{i\rightarrow \infty}
\frac{|X^{(u+1)}|}{|X^{(u)}|} = \alpha \geq 1$ (note that $\alpha=1$
is equivalent to sorting because all layers have size 1). By assigning
values in layer $u$ children from layer $u+1$, this can be seen as a
more constrained form of heap; however, unlike sorting, for any
constant $\alpha>1$, LOHs can be constructed $\in O(n)$ by performing
iterative linear time one-dimensional selection, iteratively selecting
and removing the largest layer until all layers have been partitioned.

Although selections reminiscent of LOHs may have been used previously,
formalization of rank $\alpha$ LOHs has been necessary to demonstrate
that for $1\ll\alpha\ll 2$, a combination of LOHs and soft heaps allow
generating the minimum $k$ values from $X_1+X_2+\cdot+X_m$ (where each
$X_i$ has length $n$) in $o(n\cdot m + k\cdot
m)$\cite{kreitzberg:selection}. Furthermore, efficiently constructing
an LOH of rank $\alpha$ is not trivial when $\alpha\ll 2$; after all,
$\alpha\approx 1$ results in layers of size $|X^{(1)}|=|X^{(2)}|=\cdots=1$,
indicating a sorting, which implies a runtime
$\in\Omega(n\log(n))$\cite{pennington:optimal}.

A {\tt python} implementation of a LOH is shown in
listing~\ref{alg:layer-ordered-heap}.

\subsection{Contribution in this manuscript}
The new, optimal algorithm for solving top-$k$ presented here makes
extensive use of LOHs. It is simple to implement, does not rely on
anything more complicated than linear time one-dimensional selection
(\emph{i.e.}, it does not use soft heap). Due to its simplicity and
contiguous memory access, it has fast performance in practice.

\section{Methods}
\subsection{Algorithm}
\subsubsection{Phase 0}
The algorithm first LOHifies (\emph{i.e.}, constructs a layer order heap
from) both $X$ and $Y$. This is performed by using linear time
one-dimensional selection to iteratively remove the largest remaining
layer (\emph{i.e.}, the simplest LOH construction method, which is
optimal when $\alpha\gg 1$).

\subsubsection{Phase 1}
Now layer products of the form $X^{(u)} + Y^{(v)} = X^{(u)}_1 +
Y^{(v)}_1, X^{(u)}_1 + Y^{(v)}_2, \ldots X^{(u)}_2 + Y^{(v)}_1,
\ldots$ are considered, where $X^{(u)}$ and $Y^{(v)}$ are layers of
their respective LOHs.

In phases 1--2, the algorithm initially considers only the minimum and
maximum values in each layer product: $\lfloor(u,v)\rfloor =
(\min(X^{(u)}+Y^{(v)}), (u,v), false)$, $\lceil (u,v)\rceil =
(\max(X^{(u)}+Y^{(v)}), (u,v), true)$. It is unnecessary to compute
the Cartesian product of values to build a layer product; instead,
only the minimum or maximum values in $X^{(u)}$ and $Y^{(v)}$ are
needed. Note that $false$ is used to indicate that this is the minimal
value in the layer product, while $true$ indicates the maximum value
in the layer product. Let $false=0, true=1$ so that
$\lfloor(u,v)\rfloor < \lceil(u,v)\rceil$. Scalar values can be
compared to tuples: $X_i+Y_j \leq \lceil(u,v)\rceil =
(\max(X^{(u)}+Y^{(v)}),(u,v),true) \leftrightarrow X_i+Y_j\leq
\max(X^{(u)}+Y^{(v)})$.

Binary heap $H$ is initialized to contain tuple $\lfloor (1,1)\rfloor$. A set
of all tuples in $H$ is maintained to prevent duplicates from being
inserted into $H$. The algorithm proceeds by popping the
lexicographically minimum tuple from $H$. W.l.o.g., there is not
guaranteed ordering of the form $X^{(u)}+Y^{(v)} \leq
X^{(u+1)}+Y^{(v)}$, because it may be that $\max(X^{(u)}+Y^{(v)}) >
\min(X^{(u+1)}+Y^{(v)})$; however, lexicographically, $\lfloor(u,v)\rfloor <
\lfloor(u+1,v)\rfloor, \lfloor(u,v+1)\rfloor, \lceil(u,v)\rceil$;
thus, the latter tuples need be inserted into $H$ only after
$\lfloor(u,v)\rfloor$ has been popped from $H$. $\lceil(u,v)\rceil$
tuples do not insert any new tuples into $H$ when they're popped.

Whenever a tuple of the form $\lceil(u,v)\rceil$ is popped from $H$,
the index $(u,v)$ is appended to list $q$ and the size of the layer
product $|X^{(u)}+Y^{(v)}|=|X^{(u)}|\cdot|Y^{(v)}|$ is accumulated
into integer $s$. This method proceeds until $s \geq k$.

\subsubsection{Phase 2}
Any remaining tuple in $H$ of the form $(\max(X^{(u')}+Y^{(v')}),
(u',v'), true)$ has its index $(u',v')$ appended to list $q$. $s'$ is
the total number of elements in each of these $(u',v')$ layer products
appended to $q$ during phase 2.

\subsubsection{Phase 3}
The values from every element in each layer product in $q$ is
generated. A linear time one-dimensional $k$-selection is performed on
these values and returned.

\subsection{Proof of correctness}
Lemma~\ref{thm:q-contains-minimal-k} proves that at termination all
layer products found in $q$ must contain the minimal $k$ values in
$X+Y$. Thus, by performing one-dimensional $k$-selection on those
values in phase 3, the minimal $k$ values in $X+Y$ are found.

\begin{lemma}
  \label{thm:convexity-of-popping-min}
  If $\lfloor(u,v)\rfloor$ is popped from $H$, then both $\lfloor(u-1,v)\rfloor$ (if $u>1$) and $\lfloor(u,v-1)\rfloor$ (if $v>1$) must previously have been popped from $H$. 
\end{lemma}
\begin{proof}
There is a chain of pops and insertions backwards from
$\lfloor(u,v)\rfloor$ to $\lfloor(1,1)\rfloor$. This chain must
include structures of pops of the form $\lfloor(a-1,b-1)\rfloor,
\lfloor(a,b-1)\rfloor, \lfloor(a,b)\rfloor$ or
$\lfloor(a-1,b-1)\rfloor, \lfloor(a-1,b)\rfloor,
\lfloor(a,b)\rfloor$. W.l.o.g., pops of $\lfloor(a-1,b-1)\rfloor,
\lfloor(a,b-1)\rfloor, \lfloor(a,b)\rfloor$ mean that
$\lfloor(a-1,b)\rfloor$ would be inserted into $H$ before
$\lfloor(a,b)\rfloor$, and since $\lfloor(a,b-1)\rfloor <
\lfloor(a,b)\rfloor$, it must be popped before
$\lfloor(a,b)\rfloor$. By that reasoning, $\lfloor(u-1,v)\rfloor$ and
$\lfloor(u,v-1)\rfloor$ must be popped before $\lfloor(u,v)\rfloor$. 
$\square$
\end{proof}

\begin{lemma}
  \label{thm:convexity-of-popping-max}
  If $\lceil(u,v)\rceil$ is popped from $H$, then both $\lceil(u-1,v)\rceil$ (if $u>1$) and $\lceil(u,v-1)\rceil$ (if $v>1$) must previously have been popped from $H$. 
\end{lemma}
\begin{proof}
  Inserting $\lceil(u,v)\rceil$ requires previously popping
  $\lfloor(u,v)\rfloor$. By lemma~\ref{thm:convexity-of-popping-min},
  this requires previously popping $\lfloor(u-1,v)\rfloor$ (if $u>1$)
  and $\lfloor(u,v-1)\rfloor$ (if $v>1$). These pops will insert
  $\lceil(u-1,v)\rceil$ and $\lceil(u,v-1)\rceil$ respectively. Thus,
  $\lceil(u-1,v)\rceil$ and $\lceil(u,v-1)\rceil$, which are both
  $<\lceil(u,v)\rceil$, are inserted before $\lceil(u,v)\rceil$, and
  will therefore be popped before $\lceil(u,v)\rceil$.
$\square$\end{proof}

\begin{lemma}
  \label{thm:corners-visited-in-sorted-order}
  Minimum (\emph{i.e.}, $\lfloor\cdot\rfloor$) and maximum
  (\emph{i.e.}, $\lceil\cdot\rceil$) tuples from all layer products
  will be popped from $H$ in ascending order.
\end{lemma}
\begin{proof}
Let $\lfloor(u,v)\rfloor$ be popped from $H$ and let
$\lfloor(a,b)\rfloor < \lfloor(u,v)\rfloor$. Either
w.l.o.g. $a<u,b\leq v$, or w.l.o.g. $a<u, b>v$. In the former case,
$\lfloor(a,b)\rfloor$ will be popped before $\lfloor(u,v)\rfloor$ by
applying induction to lemma~\ref{thm:convexity-of-popping-min}.

In the latter case, lemma~\ref{thm:convexity-of-popping-min} says that
$\lfloor(a,v)\rfloor$ is popped before
$\lfloor(u,v)\rfloor$. $\lfloor(a,v)\rfloor < \lfloor(a,b)\rfloor <
\lfloor(u,v)\rfloor$, meaning that $\forall v \geq r\leq b,~
\lfloor(a,r)\rfloor < \lfloor(u,v)\rfloor$. After
$\lfloor(a,v)\rfloor$ is inserted (necessarily before it is popped),
at least one such $\lfloor(a,r)\rfloor$ must be in $H$ until
$\lfloor(a,b)\rfloor$ is popped. Thus, all such $\lfloor(a,r)\rfloor$
will be popped before $\lfloor(u,v)\rfloor$.

Ordering on popping with $\lceil(a,b)\rceil<\lceil(u,v)\rceil$ is
shown in the same manner: For $\lceil(u,v)\rceil$ to be in $H$,
$\lfloor(u,v)\rfloor$ must have previously been popped. As above,
whenever $\lceil(u,v)\rceil$ is in $H$ at least one
$\lfloor(a,r)\rfloor, v \geq r\leq b$ must also be in $H$ until
$\lfloor(a,b)\rfloor$ is popped. These $\lfloor(a,r)\rfloor \leq
\lfloor(a,b)\rfloor < \lceil(a,b)\rceil < \lceil(u,v)\rceil$, and so
$\lceil(a,b)\rceil$ will be popped before $\lceil(u,v)\rceil$.

Identical reasoning also shows that $\lfloor(a,b)\rfloor$ will pop
before $\lceil(u,v)\rceil$ if $\lfloor(a,b)\rfloor<\lceil(u,v)\rceil$
or if $\lceil(a,b)\rceil<\lfloor(u,v)\rfloor$.

Thus, all tuples are popped in ascending order. 
$\square$\end{proof}

\begin{lemma}
  \label{thm:q-contains-minimal-k}
  At the end of phase 2, the layer products whose indices are found in $q$ contain the minimal $k$ values. 
\end{lemma}
\begin{proof}
  Let $(u,v)$ be the layer product that first makes $s\geq k$. There
  are at least $k$ values of $X+Y$ that are $\leq
  \max(X^{(u)}+Y^{(v)})$; this means that $\tau = \max( select(X+Y,k)
  ) \leq \max(X^{(u)}+Y^{(v)})$. The quality of the elements in layer
  products in $q$ at the end of phase 1 can only be improved by
  trading some value for a smaller value, and thus require a new value
  $<\max(X^{(u)}+Y^{(v)})$.

  By lemma~\ref{thm:corners-visited-in-sorted-order}, tuples will be
  popped from $H$ in ascending order; therefore, any layer product
  $(u',v')$ containing values $<\max(X^{(u)}+Y^{(v)})$ must have had
  $\lfloor(u',v')\rfloor$ popped before $\lceil(u,v)\rceil$. If
  $\lceil(u',v')\rceil$ was also popped, then this layer product is
  already included in $q$ and cannot improve it. Thus the only layers
  that need be considered further have had $\lfloor(u',v')\rfloor$
  popped but not $\lceil(u',v')\rceil$ popped; these can be found by
  looking for all $\lceil(u',v')\rceil$ that have been inserted into $H$
  but not yet popped.

  Phase 2 appends to $q$ all such remaining layer products of
  interest. Thus, at the end of phase 2, $q$ contains all layer
  products that will be represented in the $k$-selection of $X+Y$.
$\square$\end{proof}

A {\tt python} implementation of this method is shown in
listing~\ref{alg:cartesian-sum-select}.

\subsection{Runtime}
Theorem~\ref{thm:total-runtime-in-o-n-plus-k} proves that the total
runtime is $\in O(n+k)$.

\begin{lemma}
  \label{thm:top-left-neighbor-max-popped-before-self-min-popped}
Let $(u',v')$ be a layer product appended to $q$ during phase
2. Either $u'=1$, $v'=1$, or $(u'-1,v'-1)$ was already appended to $q$
in phase 1.
\end{lemma}
\begin{proof}
  Let $u'>1$ and $v'>1$. By
  lemma~\ref{thm:corners-visited-in-sorted-order}, minimum and maximum
  layer products are popped in ascending order. By the layer ordering
  property of $X$ and $Y$, $\max(X^{(u'-1)}) \leq \min(X^{(u')})$ and
  $\max(Y^{(v'-1)}) \leq \min(Y^{(v')})$. Thus,
  $\lceil(u'-1,v'-1)\rceil < \lfloor(u',v')\rfloor$ and so
  $\lceil(u'-1,v'-1)\rceil$ must be popped before
  $\lfloor(u',v')\rfloor$.
$\square$\end{proof}

\begin{lemma}
  \label{thm:elements-appended-to-q-in-phase-1-in-o-k}
  $s$, the number of elements in all layer products appended to $q$ in phase 1, is $\in O(k)$. 
\end{lemma}
\begin{proof}
$(u,v)$ is the layer product whose inclusion during phase 1 in $q$
  achieves $s\geq k$; therefore, $s - |X^{(u)}+Y^{(v)}| < k$. This
  happens when $\lceil(u,v)\rceil$ is popped from $H$.

If $k=1$, popping $\lceil(1,1)\rceil$ ends phase 1 with $s=1\in O(k)$.

If $k>1$, then at least one layer index is $>1$: $u>1$ or $v>1$. W.l.o.g., let $u>1$. By
lemma~\ref{thm:convexity-of-popping-min}, popping $\lceil(u,v)\rceil$
from $H$ requires previously popping
$\lceil(u-1,v)\rceil$. $|X^{(u)}+Y^{(v)}|=|X^{(u)}|\cdot|Y^{(v)}|\approx
\alpha \cdot |X^{(u-1)}|\cdot|Y^{(v)}| = \alpha \cdot |X^{(u-1)} +
Y^{(v)}|$; therefore, $|X^{(u)}+Y^{(v)}|\in
O(|X^{(u-1)}+Y^{(v)}|)$. $|X^{(u-1)} + Y^{(v)}|$ is already counted in
$s - |X^{(u)}+Y^{(v)}|<k$, and so $|X^{(u-1)} + Y^{(v)}|<k$ and
$|X^{(u)}+Y^{(v)}|\in O(k)$. $s < k+|X^{(u)}+Y^{(v)}| \in O(k)$ and
hence $s \in O(k)$.
$\square$\end{proof}

\begin{lemma}
  \label{thm:elements-appended-to-q-in-phase-2-in-o-k}
$s'$, the total number of elements in all layer products appended to $q$ in phase 2, $\in O(k)$.
\end{lemma}
\begin{proof}
Each layer product appended to $q$ in phase 2 has had
$\lfloor(u',v')\rfloor$ popped in phase 1. By
lemma~\ref{thm:top-left-neighbor-max-popped-before-self-min-popped},
either $u'=1$ or $v'=1$ or $\lceil(u'-1,v'-1)\rceil$ must have been
popped before $\lfloor(u',v')\rfloor$.

First consider when $u'>1$ and $v'>1$. Each $(u',v')$ matches to
exactly one layer product $(u'-1,v'-1)$. Because
$\lceil(u'-1,v'-1)\rceil$ must have been popped before
$\lfloor(u',v')\rfloor$, then $\lceil(u'-1,v'-1)\rceil$ was also
popped during phase 1. $s$, the count of all elements whose layer
products were inserted into $q$ in phase 1, includes $|X^{(u'-1)} +
Y^{(v'-1)}|$ but does not include $X^{(u')}+Y^{(v')}$ (the latter is
appended to $q$ during phase 2). By exponential growth of layers in
$X$ and $Y$, $|X^{(u')}+Y^{(v')}|\approx \alpha^2 \cdot
|X^{(u'-1)}+Y^{(v'-1)}|$. These $|X^{(u'-1)}+Y^{(v'-1)}|$ values were
included in $s$ during phase 1, and thus the total number of elements
in all such $(u'-1,v'-1)$ layer products is $\leq s$. Thus the sum of
sizes of all layer products $(u',v')$ with $u'>1$ and $v'>1$ that are
appended to $q$ during phase 2 is $\approx\leq\alpha^2\cdot s$. 

When either $u'=1$ or $v'=1$, the total area of layer products must be
$\in O(n)$: $\sum_{u'} |X^{(u')}+Y^{(1)}| + \sum_{v'}
|X^{(u')}+Y^{(1)}| <2n$; however, it is possible to show that
contributions where $u'=1$ or $v'=1$ are $\in O(k)$:

W.l.o.g. for $u'>1$, $\lfloor(u',1)\rfloor$ is inserted into $H$ only
when $\lfloor(u'-1,1)\rfloor$ is popped. Thus at most one
$\lfloor(u',1)\rfloor$ can exist in $H$ at any time. Furthermore,
popping $\lfloor(u',1)\rfloor$ from $H$ requires previously popping
$\lceil(u'-1,1)\rceil$ from $H$: layer ordering on $X$ implies
$\max(X^{(u'-1)})\leq\min(X^{(u')})$ and $|Y^{(1)}=1|$ implies
$\min(Y^{(1)})=\max(Y^{(1)})$, and so $\lceil(u'-1,1)\rceil =
(\max(X^{(u'-1)}+Y^{(1)}), (u'-1,1), true) < \lfloor(u',1)\rfloor =
(\min(X^{(u')}+Y^{(1)}), (u',1), false)$. Thus $\lceil(u'-1,1)\rceil$
has been popped from $H$ and counted in $s$. By the exponential growth
of layers, the contribution of all such $u'>1, v'=1$ will be $\approx
\leq \alpha\cdot s$, and so the contributions of $u'>1, v'=1$ or
$u'=1, v'>1$ will be $\approx \leq 2 \alpha\cdot s$.

When $u'=v'=1$, the layer product has area 1. 

Therefore, $s'$, the total number of elements found in layer products
appended to $q$ during phase 2, has $s'\leq (\alpha^2 + 2\alpha)\cdot
s + 1$. By lemma~\ref{thm:elements-appended-to-q-in-phase-1-in-o-k}, $s\in
O(k)$, and thus $s'\in O(k)$.
$\square$
\end{proof}

\begin{theorem}
  \label{thm:total-runtime-in-o-n-plus-k}
The total runtime of the algorithm is $\in O(n+k)$.
\end{theorem}
\begin{proof}
For any constant $\alpha>1$, LOHification of $X$ and $Y$ runs in
linear time, and so phase 0 runs $\in O(n)$.

The total number of layers in each LOH is $\approx \log_\alpha(n)$;
therefore, the total number of layer products is $\approx
\log^2_\alpha(n)$. In the worst-case scenario, the heap insertions and
pops (and corresponding set insertions and removals) will sort
$\approx 2 \log^2_\alpha(n)$ elements, because each layer product may
be inserted as both $\lfloor\cdot\rfloor or \lceil\cdot\rceil$; the
worst-case runtime via comparison sort will be $\in O(\log^2_\alpha(n)
\log(\log^2_\alpha(n))) \subset o(n)$. The operations to maintain a
set of indices on the hull have the same runtime per operation as
those inserting/removing to a binary heap, and so can be amortized
out. Thus, the runtimes of phases 1--2 are amortized out by the $O(n)$
runtime of phase 0.
  
Lemma~\ref{thm:elements-appended-to-q-in-phase-1-in-o-k} shows that
$s\in O(k)$. Likewise,
lemma~\ref{thm:elements-appended-to-q-in-phase-2-in-o-k} shows
that $s'\in O(k)$. The number of elements in all layer products in
$q$ during phase 3 is $s+s'\in O(k)$. Thus, the number of elements
on which the one-dimensional selection is performed will be $\in
O(k)$. Using a linear time one-dimensional selection algorithm, the
runtime of the $k$-selection in phase 3 is $\in O(k)$.

The total runtime of all phases $\in O(n+k+k+k)=O(n+k)$. 
$\square$

\end{proof}

\subsection{Space}
Space $\leq$ time, because each unit of work can only allocate
constant space. Thus the space usage is $\in O(n+k)$.

\section{Results}
Runtimes of the naive $O(n^2\log(n)+k)$ method, the soft heap-based
method from Kaplan \emph{et al.}, and the LOH-based method in this
paper are shown in table~\ref{table:runtimes}. The proposed approach
achieves a $>295\times$ speedup over the naive approach and
$>18\times$ speedup over the soft heap approach.

\begin{table}
\centering
\footnotesize
\begin{tabular}{ l | c c c c }

& \begin{tabular}{c} Naive \\ $n^2\log(n)+k$ \end{tabular} & \begin{tabular}{c} Kaplan \emph{et al.} \\ soft heap \end{tabular} & \begin{tabular}{c} Layer-ordered heap \\ (total=phase 0+phases 1--3) \end{tabular}\\

  \hline
$n=1000,k=250$ & 0.939 & 0.0511 & 0.00892=0.00693+0.002\\
$n=1000,k=500$ & 0.952 & 0.099 & 0.0102=0.00648+0.00374\\
$n=1000,k=1000$ & 0.973 & 0.201 & 0.014=0.00764+0.00639\\
$n=1000,k=2000$ & 0.953 & 0.426 & 0.0212=0.00652+0.0146\\
$n=1000,k=4000$ & 0.950 & 0.922 & 0.0278=0.00713+0.0206\\
\hline
$n=2000,k=500$ & 4.31 & 0.104 & 0.0194=0.0160+0.00342\\
$n=2000,k=1000$ & 4.11 & 0.203 & 0.0211=0.0139+0.00728 \\
$n=2000,k=2000$ & 4.17 & 0.432 & 0.0254=0.0140+0.0114\\
$n=2000,k=4000$ & 4.16 & 0.916 & 0.0427=0.0147+0.0280\\
$n=2000,k=8000$ & 4.13 & 2.03 & 0.0761=0.0143+0.0617\\
\hline
$n=4000,k=1000$ & 17.2 & 0.207 & 0.0507=0.0459+0.00488\\
$n=4000,k=2000$ & 17.2 & 0.422 & 0.408=0.0268+0.0141\\
$n=4000,k=4000$ & 17.1 & 0.907 & 0.0481=0.0277+0.0205\\
$n=4000,k=8000$ & 17.3 & 1.98 & 0.0907=0.0278+0.0629\\
$n=4000,k=16000$ & 17.3 & 4.16 & 0.133=0.0305+0.103\\
\end{tabular}
  \caption{Average runtimes on random uniform integer $X$ and $Y$ with $|X|=|Y|=n$. The layer-ordered heap implementation used $\alpha=2$ and resulted in $\frac{s+s'}{k} = 3.637$ on average. Individual and total runtimes are rounded to three significant figures.}
\label{table:runtimes}
\end{table}

\section{Discussion}
The algorithm can be thought of as ``zooming out'' as it pans through
the layer products, thereby passing the value threshold at which the
$k^{\text{th}}$ best value $X_i+Y_j$ occurs. It is somewhat reminiscent of
skip lists\cite{pugh:skip}; however, where a skip list begins coarse
and progressively refines the search, this approach begins finely and
becomes progressively coarser. The notion of retrieving the best $k$
values while ``overshooting'' the target by as little as possible
results in some values that may be considered but which will not
survive the final one-dimensional selection in phase 3. This is
reminiscent of ``corruption'' in Chazelle's soft heaps. Like soft
heaps, this method eschews sorting in order to prevent a runtime $\in
\Omega(n \log(n))$ or $\in \Omega(k \log(k))$. But unlike soft heaps,
LOHs can be constructed easily using only an implementation of
median-of-medians (or any other linear time one-dimensional selection
algorithm).

Phase 3 is the only part of the algorithm in which $k$ appears in the
runtime formula. This is significant because the layer products in $q$
at the end of phase 2 could be returned in their compressed form
(\emph{i.e.}, as the two layers to be combined). The total runtime of
phases 0--2 is $\in O(n)$. It may be possible to recursively perform
$X+Y$ selection on layer products $X^{(u)}+Y^{(v)}$ to compute layer
products constituting exactly the $k$ values in the solution, still in
factored Cartesian layer product form. Similarly, it may be possible
to perform the one-dimensional selection without fully inflating every
layer product into its constituent elements. For some applications, a
compressed form may be acceptable, thereby making it plausible to
remove the requirement that the runtime be $\in \Omega(k)$.

As noted in theorem~\ref{thm:total-runtime-in-o-n-plus-k}, even fully
sorting all of the minimal and maximum layer products would be $\in
o(n)$; sorting in this manner may be preferred in practice, because it
simplifies the implementation
(Listing~\ref{alg:simplified-cartesian-sum-select}) at the cost of
incurring greater runtime in practice when $k \ll n^2$. Furthermore,
listing~\ref{alg:simplified-cartesian-sum-select} is unsuitable for
online processing (\emph{i.e.}, where $X$ and $Y$ are extended on the
fly or where several subsequent selections are performed), whereas
listing~\ref{alg:cartesian-sum-select} could be adapted to those uses.

Phase 0 (which performs LOHification) is the slowest part of the
presented python implementation; it would benefit from having a
practically faster implementation to perform LOHify.

The fast practical performance is partially due to the algorithm's
simplicity and partially due to the contiguous nature of LOHs. Online
data structures like soft heap are less easily suited to contiguous
access, because they support efficient removal and therefore move
pointers to memory rather than moving the contents of the memory.

\section{Acknowledgements}
Thanks to Patrick Kreitzberg, Kyle Lucke, and Jake Pennington for
fruitful discussions and kindness.

\section{Declarations}

\subsection{Funding}
This work was supported by grant number 1845465 from the National
Science Foundation.

\subsection{Code availability}
Python source code and \LaTeX ~for this paper are available at
\url{https://bitbucket.org/orserang/selection-on-cartesian-product/}
(MIT license, free for both academic and commercial use). C++ code
(MIT license, free for both academic and commercial use) can be found
in \url{https://bitbucket.org/orserang/neutronstar/}, the world's
fastest isotope calculator.

\subsection{Conflicts of interest}
 The authors declare that they have no conflicts of interest.

\section{Supplemental information}
\subsection{Python code}
 \lstset{language=Python,
   basicstyle=\ttfamily\footnotesize,
   keywordstyle=\color{blue}\ttfamily,
   stringstyle=\color{red}\ttfamily,
   commentstyle=\color{magenta}\ttfamily,
   morecomment=[l][\color{magenta}]{\#},
   breaklines=true,
   label=alg:layer-ordered-heap
 }
\begin{lstlisting}[caption={{\tt LayerOrderedHeap.py}: A class for LOHifying, retrieving layers, and the minimum and maximum value in a layer.}]
# https://stackoverflow.com/questions/10806303/python-implementation-of-median-of-medians-algorithm
def median_of_medians_select(L, j): # returns j-th smallest value:
  if len(L) < 10:
    L.sort()
    return L[j]
  S = []
  lIndex = 0
  while lIndex+5 < len(L)-1:
    S.append(L[lIndex:lIndex+5])
    lIndex += 5
  S.append(L[lIndex:])
  Meds = []
  for subList in S:
    Meds.append(median_of_medians_select(subList, int((len(subList)-1)/2)))
  med = median_of_medians_select(Meds, int((len(Meds)-1)/2))
  L1 = []
  L2 = []
  L3 = []
  for i in L:
    if i < med:
      L1.append(i)
    elif i > med:
      L3.append(i)
    else:
      L2.append(i)
  if j < len(L1):
    return median_of_medians_select(L1, j)
  elif j < len(L2) + len(L1):
    return L2[0]
  else:
    return median_of_medians_select(L3, j-len(L1)-len(L2))

def partition(array, left_n):
  n = len(array)
  right_n = n - left_n

  # median_of_medians_select argument is index, not size:
  max_value_in_left = median_of_medians_select(array, left_n-1)

  left = []
  right = []
  for i in range(n):
    if array[i] < max_value_in_left:
      left.append(array[i])
    elif array[i] > max_value_in_left:
      right.append(array[i])
  num_at_threshold_in_left = left_n - len(left)
  left.extend([max_value_in_left]*num_at_threshold_in_left)
  num_at_threshold_in_right = right_n - len(right)
  right.extend([max_value_in_left]*num_at_threshold_in_right)
  return left, right

def layer_order_heapify_alpha_eq_2(array):
  n = len(array)
  if n == 0:
    return []
  if n == 1:
    return array
  new_layer_size = 1
  layer_sizes = []
  remaining_n = n
  while remaining_n > 0:
    if remaining_n >= new_layer_size:
      layer_sizes.append(new_layer_size)
    else:
      layer_sizes.append(remaining_n)
    remaining_n -= new_layer_size
    new_layer_size *= 2
  result = []
  for i,ls in enumerate(layer_sizes[::-1]):
    small_vals,large_vals = partition(array, len(array) - ls)
    array = small_vals
    result.append(large_vals)
  return result[::-1]

class LayerOrderedHeap:
  def __init__(self, array):
    self._layers = layer_order_heapify_alpha_eq_2(array)
    self._min_in_layers = [ min(layer) for layer in self._layers ]
    self._max_in_layers = [ max(layer) for layer in self._layers ]
    #self._verify()

  def __len__(self):
    return len(self._layers)

  def _verify(self):
    for i in range(len(self)-1):
      assert(self.max(i) <= self.min(i+1))

  def __getitem__(self, layer_num):
    return self._layers[layer_num]

  def min(self, layer_num):
    assert( layer_num < len(self) )
    return self._min_in_layers[layer_num]

  def max(self, layer_num):
    assert( layer_num < len(self) )
    return self._max_in_layers[layer_num]

  def __str__(self):
    return str(self._layers)

\end{lstlisting}

 \lstset{language=Python,
   basicstyle=\ttfamily\footnotesize,
   keywordstyle=\color{blue}\ttfamily,
   stringstyle=\color{red}\ttfamily,
   commentstyle=\color{magenta}\ttfamily,
   morecomment=[l][\color{magenta}]{\#},
   breaklines=true,
   label=alg:cartesian-sum-select
 }
\begin{lstlisting}[caption={{\tt CartesianSumSelection.py}: A class for efficiently performing selection on $X+Y$ in $\Theta(n+k)$ steps.}]
from LayerOrderedHeap import *
import heapq

class CartesianSumSelection:
  def _min_tuple(self,i,j):
    # True for min corner, False for max corner
    return (self._loh_a.min(i) + self._loh_b.min(j), (i,j), False)

  def _max_tuple(self,i,j):
    # True for min corner, False for max corner
    return (self._loh_a.max(i) + self._loh_b.max(j), (i,j), True)

  def _in_bounds(self,i,j):
    return i < len(self._loh_a) and j < len(self._loh_b)

  def _insert_min_if_in_bounds(self,i,j):
    if not self._in_bounds(i,j):
      return

    if (i,j,False) not in self._hull_set:
      heapq.heappush(self._hull_heap, self._min_tuple(i,j))
      self._hull_set.add( (i,j,False) )

  def _insert_max_if_in_bounds(self,i,j):
    if not self._in_bounds(i,j):
      return

    if (i,j,True) not in self._hull_set:
      heapq.heappush(self._hull_heap, self._max_tuple(i,j))
      self._hull_set.add( (i,j,True) )

  def __init__(self, array_a, array_b):
    self._loh_a = LayerOrderedHeap(array_a)
    self._loh_b = LayerOrderedHeap(array_b)

    self._hull_heap = [ self._min_tuple(0,0) ]
    # False for min:
    self._hull_set = { (0,0,False) }
    
    self._num_elements_popped = 0
    self._layer_products_considered = []

    self._full_cartesian_product_size = len(array_a) * len(array_b)

  def _pop_next_layer_product(self):
    result = heapq.heappop(self._hull_heap)
    val, (i,j), is_max = result
    self._hull_set.remove( (i,j,is_max) )

    if not is_max:
      # when min corner is popped, push their own max and neighboring mins
      self._insert_min_if_in_bounds(i+1,j)
      self._insert_min_if_in_bounds(i,j+1)
      self._insert_max_if_in_bounds(i,j)
    else:
      # when max corner is popped, do not push
      self._num_elements_popped += len(self._loh_a[i]) * len(self._loh_b[j])
      self._layer_products_considered.append( (i,j) )

    return result

  def select(self, k):
    assert( k <= self._full_cartesian_product_size )

    while self._num_elements_popped < k:
      self._pop_next_layer_product()
    
    # also consider all layer products still in hull
    for val, (i,j), is_max in self._hull_heap:
      if is_max:
        self._num_elements_popped += len(self._loh_a[i]) * len(self._loh_b[j])
        self._layer_products_considered.append( (i,j) )
      
    # generate: values in layer products

    # Note: this is not always necessary, and could lead to a potentially large speedup.
    candidates = [ val_a+val_b for i,j in self._layer_products_considered for val_a in self._loh_a[i] for val_b in self._loh_b[j] ]
    print( 'Ratio of total popped candidates to k: {}'.format(len(candidates) / k) )
    k_small_vals, large_vals = partition(candidates, k)
    return k_small_vals
\end{lstlisting}

 \lstset{language=Python,
   basicstyle=\ttfamily\footnotesize,
   keywordstyle=\color{blue}\ttfamily,
   stringstyle=\color{red}\ttfamily,
   commentstyle=\color{magenta}\ttfamily,
   morecomment=[l][\color{magenta}]{\#},
   breaklines=true,
   label=alg:simplified-cartesian-sum-select
 }
\begin{lstlisting}[caption={{\tt SimplifiedCartesianSumSelection.py}: A simplified implementation of Listing~\ref{alg:cartesian-sum-select}. This implementation is slower when $k \ll n^2$; however, it has the same asymptotic runtime for any $n$ and $k$: $\Theta(n+k)$.}]
from LayerOrderedHeap import *

class SimplifiedCartesianSumSelection:
  def _min_tuple(self,i,j):
    # True for min corner, False for max corner
    return (self._loh_a.min(i) + self._loh_b.min(j), (i,j), False)

  def _max_tuple(self,i,j):
    # True for min corner, False for max corner
    return (self._loh_a.max(i) + self._loh_b.max(j), (i,j), True)

  def __init__(self, array_a, array_b):
    self._loh_a = LayerOrderedHeap(array_a)
    self._loh_b = LayerOrderedHeap(array_b)

    self._full_cartesian_product_size = len(array_a) * len(array_b)

    self._sorted_corners = sorted([self._min_tuple(i,j) for i in range(len(self._loh_a)) for j in range(len(self._loh_b))] + [self._max_tuple(i,j) for i in range(len(self._loh_a)) for j in range(len(self._loh_b))])

  def select(self, k):
    assert( k <= self._full_cartesian_product_size )
    
    candidates = []
    index_in_sorted = 0
    num_elements_with_max_corner_popped = 0
    while num_elements_with_max_corner_popped < k:
      val, (i,j), is_max = self._sorted_corners[index_in_sorted]
      new_candidates = [ v_a+v_b for v_a in self._loh_a[i] for v_b in self._loh_b[j] ]
      if is_max:
        num_elements_with_max_corner_popped += len(new_candidates)
      else:
        # Min corners will be popped before corresponding max corner;
        # this gets a superset of what is needed (just as in phase 2)
        candidates.extend(new_candidates)
      index_in_sorted += 1
    
    print( 'Ratio of total popped candidates to k: {}'.format(len(candidates) / k) )
    k_small_vals, large_vals = partition(candidates, k)
    return k_small_vals
\end{lstlisting}

\end{document}